\documentclass[11pt,a4paper]{article}

\usepackage{amsfonts}
\usepackage{amsthm}
\usepackage{a4wide}
\usepackage{longtable}

\def\it{\textit} 
\def\bf{\textbf} 
 
\def\ul{\underline}

\def\mc{\mathcal}
\def\mb{\mathbf}
\def\bb{\mathbb}
\def\mr{\mathrm}

\def\Q{\mathrm{Q}}
\def\O{\mathrm{opt}}
\def\ds{\displaystyle}
\def\<_m{<_{\mathrm{m}}}

\newcommand{\be}{\begin{equation}}
\newcommand{\ee}{\end{equation}}
\newcommand{\benum}{\begin{enumerate}}
\newcommand{\eenum}{\end{enumerate}}
\newcommand{\bit}{\begin{itemize}}
\newcommand{\eit}{\end{itemize}}

\newtheorem{thom}{Theorem}
\newtheorem{lemma}[thom]{Lemma}
\newtheorem{rem}[thom]{Remark}
\newtheorem{prop}[thom]{Proposition}
\newtheorem{conj}[thom]{Conjecture}
\newtheorem{assum}[thom]{Assumption}
\newtheorem{prob}[thom]{Problem}
\newtheorem{defn}[thom]{Definition}

\begin{document}

\title{Lower Bounds on Syntactic Logic Expressions for Optimization
Problems and Duality using Lagrangian Dual to characterize optimality
conditions}

\author{Prabhu Manyem \\
Department of Mathematics \\
Shanghai University \\
Shanghai 200444, China. \\
Email: \texttt{prabhu.manyem@gmail.com}}

\maketitle

\begin{abstract}
We show that simple syntactic expressions such as existential second
order (ESO) universal Horn formulae can express NP-hard optimisation
problems.
There is a significant difference between the expressibilities of
decision problems and optimisation problems.
This is similar to the difference in computation times for the two
classes of problems; for example, a 2SAT Horn formula can be satisfied
in polynomial time, whereas the optimisation version in NP-hard.
It is known that all polynomially solvable decision problems can be
expressed as ESO universal ($\Pi_1$) Horn sentences in the presence of a
successor relation.
We show here that, on the other hand, if $P \neq NP$, optimisation
problems defy such a characterisation, by demonstrating that even a
$\Pi_0$ (quantifier free) Horn formula is unable to guarantee polynomial
time solvability.
Finally, by connecting concepts in optimisation duality with those in
descriptive complexity, we will show a method by which optimisation
problems can be solved by a single call to a ``decision" Turing machine,
as opposed to multiple calls using a classical binary search setting.
\end{abstract}


\thispagestyle{empty}

\section{Notation and Definitions}\label{sec:notation}

We first acquaint the reader with some background in Finite Model Theory
and Descriptive Complexity, and how they relate to optimisation.
For further reference, please see the books by Ebbinghaus and Flum
\cite{EF99} and Immerman \cite{immerman}.

\begin{defn}
\cite{cjtcs08}
A \bf{P-optimisation} problem $Q$ is a tuple
$Q = \{I_\Q, F_\Q, f_\Q, opt_\Q\}$, where
\begin{enumerate}
\item[(i)] $I_\Q$ is a set of instances to $\Q$,

\item [(ii)]
$F_\Q (I)$ is the set of feasible solutions to instance $I$,

\item[(iii)]
$f_\Q (I, S)$ is the {\em objective function} value to a solution $S \in
F_\Q(I)$ of an instance $I \in I_\Q$.
~It is a function 
$f:\bigcup_{I \in I_\Q}  [\{I\} \times F_\Q(I)] \rightarrow \bb{R}^+_0$ 
(non-negative
reals)\footnote{Of course, when it comes to computer representation,
rational numbers will be used.}, computable in time polynomial in the
size $|A|$ of the domain $A$ of $I$\footnote{Strictly speaking, we should use
$|I|$ here, where $|I|$ is the length of the representation of $I$.
~However, $|I|$ is polynomial in $|A|$, hence we can use  $|A|$.},

\item[(iv)]
For an instance $I \in I_\Q$, $opt_\Q (I)$ is either the minimum or
maximum possible value that can be obtained for the objective function,
taken over all feasible solutions in $F_\Q(I)$.

$\displaystyle
opt_\Q (I) = \max_{S \in F_\Q(I)} f_\Q (I, S)$ (for P-maximisation
problems),

$\displaystyle
opt_\Q (I) = \min_{S \in F_\Q(I)} f_\Q (I, S)$ (for P-minimisation
problems),

\item[(v)]
The following decision problem is in the class $\mb{P}$:
\it{Given an instance $I$ and a non-negative constant $k$, is there a
feasible solution $S \in F_\Q (I)$, such that 
$f_\Q (I, S) \ge k$ (for a P-maximisation problem), or 
$f_\Q (I, S) \le k$ (in the case of a P-minimisation problem)?}

And finally,

\item[(vi)]
An optimal solution $S_{\mr{opt}} (I)$ for a given instance $I$ can be
computed in time polynomial in $|I|$, where
$\displaystyle opt_\Q (I) = f_\Q (I, S_{\mr{opt}} (I))$.
% (LET ME LEAVE THIS POINT HERE FOR THE TIME BEING.)
\end{enumerate}

The set of all such $\mb{P}$-optimisation problems is
the $\mb{P_{opt}}$ class.
\label{def:pOptProblem}
\end{defn}

[Note: Some researchers dispute the presence of item (vi) above, whereas
some agree to its presence \cite{brentGuruMahajan}.
My argument is, for a P-optimisation problem, the optimal solution must
be computable in polynomial time.  Think of it as a function Turing
machine that completes all 3 tasks in polynomial time: reading the input,
computing the optimal solution, and writing the output (solution).]

A similar definition, for {NP-optimisation} problems, appeared in
Panconesi and Ranjan (1993) \cite{pancoRanjan}:
\begin{defn}
An \bf{NP-optimisation} problem is defined as follows.
Points \it{(i)-(iv)}
in Definition \ref{def:pOptProblem} above apply to NP-optimisation problems,
whereas \it{(vi)} does not.  Point \it{(v)} is modified as follows:

(v) The following decision problem is in $\mb{NP}$:
\it{Given an instance $I$ and a non-negative constant $k$, is there a feasible
solution $S \in F_\Q (I)$, such that $f_\Q (I, S) \ge k$ (for an
NP-maximisation problem), or $f_\Q (I, S) \le k$ (in the case of an
NP-minimisation problem) ?}

The set of all such $\mb{NP}$-optimisation problems is
the $\mb{NP_{opt}}$ class, and
$\mb{P_{opt}} \subseteq \mb{NP_{opt}}$.
\label{def:npOptProblem}
\end{defn}

\begin{defn}
\cite{KT94}
An optimisation problem $Q$ is said to
be  \bf{polynomially bound} if the value of an optimal solution
to every instance $I$ of $Q$ is bound by a polynomial
in the size of $I$. ~In other words, for every problem $Q$, there exists
a polynomial $p_\Q$, such that
\begin{equation}
opt_{Q} (I) \leq p_\Q (\vert I \vert),
\label{polyBoundDef}
\end{equation}
for {every} instance $I$ of $Q$.
~$\mb{P_{opt}^{pb}}$ ($\mb{NP_{opt}^{pb}}$) is the set of polynomially-bound
$\mb{P}$-optimisation ($\mb{NP}$-optimisation) problems.
Naturally,
$\mb{P_{opt}^{pb}} \subseteq \mb{P_{opt}}$ and
$\mb{NP_{opt}^{pb}} \subseteq \mb{NP_{opt}}$.
\label{def:polyMax}
\end{defn}

\begin{defn}
\cite{EF99}
\bf{First order (FO) logic} consists of a
{\em vocabulary} (alias signature) $\sigma$, and 
{\em structures} defined on the vocabulary.

In its simplest form, a \it{vocabulary} consists of a set of variables, a
set of constants, and a
set of relation symbols $R_j$ $(1 \leq j \leq J_1)$, each of arity $r_j$,
where $J_1$ is a finite positive integer.

A \it{structure} $M$ consists of a universe $U$ whose elements are the values
that variables can take.
$M$ also instantiates each relation symbol $R_j \in \sigma$ with tuples
from $U^{(r_j)}$.

When a structure $\mc{A}$ satisfies a formula $\phi$ (written as
$\mc{A} \models \phi$), $\mc{A}$ is said to {\em model} $\phi$, or,
$\mc{A}$ is a {\em model} for $\phi$. 
\label{def:FOL}
\end{defn}

For example, a vocabulary $\sigma_G$ in graph theory may comprise a set
of variables, two constants $s$ and $t$, and a single binary relation
symbol, $E$.

A structure $\mb{G}$ in graph theory may have the set of vertices $G =
\{1, 2, \cdots 10\}$ as its universe, in addition to the constants $s$
and $t$ (assuming that the graph has 12 vertices), and the relation $E$,
where $E(i,j)$ is true iff $(i,j)$ is an edge in the graph $\mb{G}$.
~The vertices $s$ and $t$ are two special vertices in $G$,
and may represent the origin and destination, respectively.

A structure represents an instance of an optimisation problem.  

We give a definition of FO formulae in their simplest form:
\begin{defn}
\cite{EF99}
\bf{FO formulae}.
\newline
(1) If $x_1$, $x_2$, $\cdots$, $x_k$ are variables or constants and $R$
is a $k$-ary relation symbol, then $R (x_1, x_2, \cdots, x_k)$ and
$x_i = x_j$ are formulae;
\newline
(2) If $\phi$ is a formula, then so is $\neg \phi$;
\newline
(3) If $\phi$ and $\psi$ are formulae, then so are $\phi \vee \psi$ and 
$\phi \wedge \psi$; and finally,
\newline
(4) If $\phi$ is a formula and $x$ is a FO variable, then
$\exists x \phi$ and $\forall x \phi$ are also formulae.
\label{def:FO_formula}
\end{defn}

\begin{defn}
\cite{papa}
We obtain \bf{second order logic} by augmenting first order logic with
second order (SO) variables. 
The SO variables are \it{relation symbols} $S_j$ ($1 \le j \le J_2$,
where $J_2$ is a finite positive integer),
defined over first order variables.
\label{def:SOL}
\end{defn}

As an example, if the underlying FO vocabulary is $\sigma_G$, then an SO
variable $P(x_1, x_2)$ can signify that a path exists from vertex $x_1$
to vertex $x_2$.

\begin{defn}
\cite{EF99}
For a formula to be in prenex normal form \bf{(PNF)},
all quantifiers appear at the beginning, followed by a
quantifier-free formula.
A \bf{$\mb{\Pi_1}$ \textnormal{(}$\mb{\Sigma_1}$\textnormal{)}
first order formula in PNF} only has
universal (existential) quantifiers that range over first order variables.
A \bf{$\mb{\Pi_2}$ \textnormal{(}$\mb{\Sigma_2}$\textnormal{)}
formula in PNF} is one that has the following form:
\begin{equation}
\phi  \equiv
\forall x_1 \cdots \forall x_a ~
\exists y_1 \cdots \exists y_b ~
\eta ~~~~ 
(\phi \equiv
\exists y_1 \cdots \exists y_b ~
\forall x_1 \cdots \forall x_a ~
\eta )
\label{Pi2Def}, 
\end{equation}
where $\eta$ is quantifier-free, 
the $x$'s and $y$'s are first-order variables and
$a, b \geq 1$.
\label{def:pie1pie2sigma1sigma2}
\end{defn}

The following definition is well known, going back to the 1950's; see for
example, \cite{gradel91}.
\begin{defn}
A \bf{Horn clause} is a disjunction of one or more literals, at most one
of which is positive. For example, $x$, $\neg y$ and
($\neg x \lor y$) are all Horn clauses, whereas ($x \lor y$) is not.
\label{def:hornClause}
\end{defn}

\begin{defn}
\cite{FI2008, immerman}
An \bf{existential second-order (ESO) Horn} expression is of the form
$\exists \mb{S} \psi$, where $\psi$ is a
first order formula, and $\mb{S} = (S_1, ~ \cdots ~ S_p)$ is a sequence
of predicate symbols not in the vocabulary of $\psi$.
The formula $\psi$ is written in $\Pi_1$ form as
\begin{equation}
\psi \equiv \forall x_1 \forall x_2 \cdots \forall x_k \eta \equiv \forall
\mb{x} ~ \eta,
\label{hornDef}
\end{equation}
where $\eta$ is a conjunction of Horn clauses ($\eta$ is, of
course, quantifier-free), and $x_i$ $(1 \leq i \leq k)$ are first
order variables.  Each clause in $\eta$ contains at most one positive
occurrence of any of the second order predicates $S_i$ ($1 \leq i \leq
p$).

A general \bf{ESO} formula is the same as an ESO Horn expression, except
that $\eta$ can now be any quantifier-free first order formula.
% $\hfill \rule{2.0mm}{2.0mm}$
\label{def:ESOhorn}
\end{defn}

\begin{defn}
\cite{EF99}
A \bf{successor relation} $succ (a,b)$, where $a \not= b$, denotes that \\
(i) $a$ immediately precedes $b$ (or $b$ immediately succeeds $a$) in $A$, 
where $A$ = universe of a structure \bf{A}, \\
(ii) $\forall c \in A$, where $a$, $b$ and $c$ are distinct,
$\neg succ (a,c) \wedge \neg succ (c,b)$, and \\
(iii) $\forall c \in A, \neg succ (c,c)$.

We assume that the vocabulary contains two constants, $min$ and $max$,
to represent the first and last elements in the universe respectively.
That \it{min} and \it{max} are the first and last elements respectively
can be expressed by the following two sentences:
\begin{equation}
\exists min ~ \forall x ~ \neg succ(x,min) ~~~\mbox{and} ~~~
\exists max ~ \forall x ~ \neg succ(max,x).
\end{equation}
\label{def:successor}
\end{defn}

Informally, $b$ occurs ``next" to $a$ in $A$, according to the above
definition --- $a$ and $b$ appear adjacent to each other in the input.  A
\it{successor relation} is different from a \it{linear order}
\cite{gradel92}.
A linear order (also known as a \it{total order}) is a binary relation
defining a sequence for every pair of elements in the domain.

\begin{defn}
(\cite{gradelAnd7others}, Chapter 3)
A first order sentence $\psi$ of vocabulary $\tau \cup \{<\}$ is said to
be \bf{order invariant} on a class $\mc{K}$ of $\tau$-structures, if its
truth on any structure in $\mc{K}$ does not depend on the choice of the
linear ordering $<$.
~That is, for any structure $\mc{U} \in \mc{K}$ and a pair of linear
orderings $<_1$  and $<_2$, it is true that
$(\mc{U}, <_1) \models \psi$ if and only if 
$(\mc{U}, <_2) \models \psi$.
\end{defn}
Henceforth, when we deal with ordered structures, we assume that the
sentences are order invariant. 

\begin{prob}
\bf{MaxHorn2Sat} \cite{jaumard87}.
\newline
{\em Given}.  A set of clauses $c_i$, $1 \leq i \leq n$.  Each clause
$c_i$ is one of the following:
(i) a Boolean variable $x_j$,
(ii) its negation,  $\neg x_j$,
(iii) $x_j \vee \neg x_k$, or
(iv) $\neg x_j \vee \neg x_k$.
\newline
{\em To Do}.  Assign truth values to the $x_i$'s such that the number
of satisfied clauses is maximised.
\label{prob:maxhorn2sat}
\end{prob}
Informally, an instance of MaxHorn2Sat consists of a formula in
conjunctive normal form (CNF), where each clause is Horn, and each
clause contains at most two literals.
(Such a formula is also known as a quadratic Horn formula.)
The problem is to maximise the number of satisfiable clauses.
The decision version of this problem is NP-complete \cite{jaumard87}.

\begin{defn}
\cite{papaYanna91, KT94}
\bf{MAX} $\mb{\Pi_0}$ is the class of maximisation problems whose
optimal solution value to an instance \bf{A} of a Problem $Q$ can be
represented as
\begin{equation}
opt_{\Q} (\mb{A}) = \max_\mb{S} \vert \{\mb{w}: (\mb{A},
\mb{S}, \mb{w}) \models  ~ \phi \}
\vert
\label{maxPi_0_Def}
\end{equation}
where $\phi = \alpha (\mb{w}, \mb{S})$ is a quantifier-free first-order
formula.
However, if $\phi$ is of the form $\ds \exists \mb{x}
~ \alpha(\mb{x}, \mb{w}, \mb{S})$ where $\alpha$ is quantifier-free,
such optimisation problems fall in the \bf{MAX NP} class, also called
the \bf{MAX} $\mb{\Sigma_1}$ class.
\label{def:maxPi0andMaxNP}
\end{defn}
See references \cite{KT94}, \cite{KT95} and \cite{papaYanna91} for an
explanation and hierarchy results for these classes.

% (Note to Referees: I am unsure if Definitions \ref{def:ESOhorn} and
% \ref{def:maxPi0andMaxNP} are necessary for a Logic journal; they seem more
% appropriate for an Algorithms or Complexity journal?)

\begin{table}
{\begin{tabular}{|l|p{330pt}|}
\hline
ESO logic &  Existential second order logic.  \\
\hline
FO logic &  First order logic.             \\
\hline
$\bf{A}  $  & a
structure defined over a signature $\sigma$;
           \bf{A} captures an instance of an optimisation problem. \\
\hline
 $\eta  $    & a quantifier-free first order (FO) formula, which is a
conjunction of \it{Horn} clauses.
(Recall that a Horn clause contains at most one positive literal.)  \\
\hline
$\mb{x}  $  &  an $m-$tuple of FO variables. \\
\hline
$\mb{S}  $  &  a sequence of predicate symbols or second order (SO) variables;\\
            &  \bf{S} captures a solution to the optimisation problem.
            \\
\hline
 $\mb{P_{opt}}$ ($\mb{NP_{opt}}$)    &
       $\mb{P}$-optimisation
       ($\mb{NP}$-optimisation) problems.  See Definition
             \ref{def:pOptProblem} (\ref{def:npOptProblem}). \\
\hline
       &    \\ [-10pt]
 $\mb{P_{opt}^{pb}}$ ($\mb{NP_{opt}^{pb}}$)    &
       Polynomially bound $\mb{P}$-optimisation
       ($\mb{NP}$-optimisation) problems.  See Definition
             \ref{def:polyMax}. \\
\hline
PNF      &  Prenex Normal Form.             \\
\hline
Sentence &  A logic formula in which every variable that appears is
            quantified. \\
\hline
\end{tabular}}
\caption{Notation}
\label{tab:defines}
\end{table}

\section{Literature Review and Our Contributions}

In a recent paper, we \cite{cjtcs08} proved the following:
\begin{thom}
Let $\sigma$ be a signature which contains a successor relation.
Let $Q$ be an optimisation problem, with finite structures $\mb{A}$ (over
$\sigma$) as instances to $Q$.
~If $Q \in \mb{P_{opt}^{pb}}$, then the value of an optimal solution to
an instance $\mb{A}$ of $Q$ can be represented by
\begin{equation}
\O_{\Q} (\mb{A}) = \O_\mb{S} \vert \{\mb{w}: (\mb{A},
\mb{S}, \mb{w}) \models \forall \mb{x} ~ \eta(\mb{w}, \mb{x}, \mb{S}) \}
\vert
\label{eq:maxDef}
\end{equation}
where $\mb{x}$, $\mb{A}$, $\mb{S}$ and $\eta$ are defined in
Table \ref{tab:defines}, and $\O \in \{\max, \min\}$.
The Horn condition in the formula $\eta$ applies only to the second
order predicates in \bf{S}, not to first order predicates.
\label{thom:hornMax}
$\hfill \rule{2.0mm}{2.0mm}$
\end{thom}

The converse of Theorem \ref{thom:hornMax} can be stated as:
\begin{prop}
\it{If the optimal solution value to an optimisation
problem $Q$ can be represented as in (\ref{eq:maxDef}), then
$Q$ belongs to the class $\mb{P_{opt}^{pb}}$}.
\label{prop:converseProp}
$\hfill \rule{2.0mm}{2.0mm}$
\end{prop}

Proposition \ref{prop:converseProp} (which deals with maximisation and
minimisation) has been shown to be false by Gate and Stewart
\cite{gateStewart} (Theorem \ref{thom:gateStewart} of this paper).
The maximisation part can be cast as in Problem \ref{prob:converseDef}
below.

For this problem, we find it more convenient to use a \it{new framework}
which was first described in \cite{KT95}.
It slightly differs from the older framework in \cite{KT94} (and used in
Theorem \ref{thom:hornMax}); the tuples \bf{w} that count towards the
objective function are now part of a new second order predicate, $S_0$.
The connection between the two frameworks is as follows:
$\ds \mb{w} \in S_0$ in the new framework iff 
$\ds (\mb{A}, \mb{S}, \mb{w}) \models \forall \mb{x} ~ \eta(\mb{w},
\mb{x}, \mb{S})$ in (\ref{eq:maxDef}).

\begin{prob}
\bf{Syntactic maximisation w.r.t a universal Horn F.O. formula} $\mb{\phi}$
\cite{FI2008}.
% $\mb{\Pi_n}$ ($\mb{\Sigma_n}$).
\newline
{\em \ul{Given}}.  (i) A structure $\mb{A}$, over
an appropriate signature $\sigma$ which contains a successor relation;
\newline
(ii) a sequence of second order variables $\mb{S} = \{S_1, \cdots,
S_p\}$ where each $S_i$ is of arity $r_i$ ($1 \le i \le p$);
\newline
(iii)
a tuple \bf{w} = ($w_1$, $w_2$, $\cdots$, $w_{r_0}$) of first order
variables of arity $r_0$;
and
\newline
(iii)
a first order universal Horn formula $\phi(\mb{w}, \mb{A}, \mb{S})$.
% a ${\Pi_n}$ (${\Sigma_n}$) formula $\phi$.

{\em \ul{To Do}}.
For $0 \le i \le p$, assign truth values to each $S_i$
such that $\phi$ is satisfied and 
$\vert \{\mb{w}: (\mb{w}, \mb{A}, \mb{S}) \models ~ \phi \} \vert$ 
is maximised.
\label{prob:converseDef}
$\hfill \rule{2.0mm}{2.0mm}$
\end{prob}

In other words, the goal is to maximise the number of tuples \bf{w} that
satisfy $\phi$;
that is, to achieve the maximum value for $opt_{} (\mb{A})$:
\begin{equation}
\O_{} (\mb{A}) = \max_{\mb{S}}
\vert \{\mb{w} \in A^{r_0}: (\mb{w}, \mb{A}, \mb{S}) \models ~ \phi \} \vert.
\label{syntMax}
\end{equation}
~If $A$ is the domain of \bf{A}, then
$S_i \subseteq A^{r_i}$, $1 \le i \le p$.

Each $S_i$ above is of the form $S_i (z_1, z_2, \cdots, z_{r_i})$,
where each $z_j$ can take any value in the domain of \bf{A}.
~First order variables are those that can be assigned values from the
domain of \bf{A}.
~A first order formula is one that contains no second-order (SO) variables.
The SO variables are quantified by SO quantifiers in an ESO formula.

\ul{Example}: Let the domain of a structure \bf{A} be 
\{$a$, $b$, $c$\}.
Let the arity of the SO variable $S_1$ be two.
Thus the nine possible tuples of $S_1$ are
($a$, $a$), ($a$, $b$), ($b$, $a$), ($b$, $b$), ($a$, $c$), ($c$, $a$),
($c$, $c$), ($b$, $c$), and ($c$, $b$).
The task is to assign truth values to each of these nine tuples of $S_1$;
and similarly for the other SO variables $S_2$, $\cdots$, $S_p$, such that 
the cardinality of the set
 $\{\mb{w} \in A^{r_0}: (\mb{w}, \mb{A}, \mb{S}) \models ~ \phi \}$ is 
maximised.

\begin{prob}
\bf{The decision version of Problem \ref{prob:converseDef}}.
% $\mb{\Pi_n}$ ($\mb{\Sigma_n}$).
\newline
{\em \ul{Given}}.  (i), (ii), (iii): Same as in Problem
\ref{prob:converseDef}; and
\newline
(iv) a constant $K$ which is a positive integer.

{\em \ul{To Do}}.
For $1 \le i \le p$, assign truth values to each $S_i$
such that $\phi$ is satisfied and 
$\vert \{\mb{w} \in A^{r_0}: (\mb{w}, \mb{A}, \mb{S}) \models ~ \phi \}
\vert \ge K$.
\label{prob:converseDefDecision}
\end{prob}

When it comes to Turing machine input, $K$ will be encoded in binary, as
usual.
But the more important question is, how is $K$ presented? Is it part of
the first order structure?
$K$ being a constant in the signature doesn't make sense, as it will be
the same for all instances.
A different $K$ should mean a different instance.
Hence it has to be a part of the domain (universe).
We can achieve this by letting the domain have different parts/sections.
For example, for a minimum spanning tree problem, the domain of the input
structure will consist of three parts: (i) the vertices, (ii) the edge
weights, and (iii) the bound $K$ on the objective function.
To our knowledge, this issue (bound on the objective function) has not been
addressed in the Descriptive Complexity literature so far.

The universal Horn formula $\phi$ in (\ref{syntMax}) can be written as
$\phi \equiv \forall \mb{x} ~ \eta(\mb{w}, \mb{x}, \mb{S})$ where $\eta$
is a quantifier-free conjunction of Horn clauses as in (\ref{eq:maxDef}).
Then for $1 \le i \le p$, the problem is to assign truth values to each
$S_i$, such that the number of tuples $\mb{w}$ that satisfy $\forall
\mb{x} ~ \eta(\mb{w}, \mb{x}, \mb{S})$ is maximised.
(As in Theorem \ref{thom:hornMax}, the Horn condition in the formula
$\eta$ applies only to the second order predicates in \bf{S}.
~This is because, the FO predicates are part of the input and hence their
truth values can be substituted, whereas the SO predicates are the
unknowns.)

Due to difficulties in computing the optimal solution value for a
general maximisation problem in $\mb{P_{opt}^{pb}}$, Bueno and Manyem
\cite{FI2008} made the following conjecture:
\begin{conj}
The optimal value for an instance \bf{A} of a maximisation problem, as
measured in (\ref{syntMax}), cannot be computed in polynomial time
by a deterministic Turing machine using syntactic (logic based) techniques.
We need optimisation algorithms that exploit the particular problem
structure.
\label{syntacticImpossible}
\end{conj}

For problems in the $\mb{NP_{opt}^{pb}}$ class, Kolaitis and Thakur
\cite{KT94} gave a precise characterisation (an ``if and only if"
result).  For problems in the $\mb{P_{opt}^{pb}}$ class, Conjecture
\ref{syntacticImpossible} predicts that Proposition \ref{prop:converseProp}
is false, and hence the ``partial characterisation" in Theorem
\ref{thom:hornMax} is one-way.

Gate and Stewart \cite{gateStewart} settled Conjecture
\ref{syntacticImpossible} with a Yes answer.
The decision version of MaxHorn2Sat (see the definition in Problem
\ref{prob:maxhorn2sat}) is known to be NP-complete \cite{jaumard87}.
Gate and Stewart were able to show a polynomial time reduction from
the decision version of MaxHorn2Sat to Problem
\ref{prob:converseDefDecision}, thus proving that 
\begin{thom}
Problem \ref{prob:converseDefDecision} is NP-hard.
\label{thom:gateStewart}
$\hfill \rule{2.0mm}{2.0mm}$
\end{thom}
\bf{Corollary}.
 Proposition \ref{prop:converseProp} is false.

In other words, the authors in \cite{gateStewart} essentially showed
that just because the optimal solution value to an optimisation problem
can be expressed in the form in (\ref{syntMax}) does not necessarily mean
that the problem is polynomially solvable; it may be NP-hard.

\subsection{Our contribution}
Here we prove a stronger negative result.  Notice that the first order
part in (\ref{syntMax}) is in $\Pi_1$ Horn form (universal Horn).  One
would expect that if we simplify the expression from $\Pi_1$ Horn to
$\Pi_0$ Horn (that is, a quantifier-free Horn formula), we should be
able to guarantee polynomial time solvability.

Unfortunately this is not the case.  We will show below that even a
quantifier-free Horn expression is unable to guarantee polynomial time
solvability.  We show this by exhibiting such an expression for an
NP-hard problem, MaxHorn2Sat.

\it{Difference between decision problems and optimisation problems}.
It is well known that if a decision problem can be
expressed as a universal ($\Pi_1$) Horn sentence in existential
second-order (ESO) logic, the problem is polynomially solvable
(see Theorem 9.32 in \cite{immerman}).
However, as we have stated above, optimisation problems differ
significantly from decision problems in their behaviour.

In Section \ref{sec:rescue}, we use optimisation duality (using the
Lagrangian Dual) to characterise optimality conditions; that section
also describes the conditions under which a single call to a ``decision
machine" (a Turing machine that solves decision problems) can obtain
optimal solutions, rather than using multiple calls to a decision machine
in a classical binary search setting.

\section{A Syntactic Expression for MaxHorn2Sat}

In this section, we will show below that when it comes to maximisation,
quantifier-free Horn expressions are unable to guarantee polynomial time
solvability.
Or, looking at this in a positive sense, quantifier-free Horn
expressions are also able to express NP-hard maximisation problems.
We show this by exhibiting such an expression for an NP-hard problem,
MaxHorn2Sat.

We need instances at two different levels.  Let us make this more clear
with an example.
Suppose we are given a MaxHorn2Sat instance (formula) such as \bf{M}
$\equiv (z_1 \vee \neg z_2) \wedge (z_3) \wedge (\neg z_3 \vee \neg
z_1)$.

The variables in this instance are $Z = \{z_1, z_2, z_3\}$, and a
structure \bf{B} maps $Z$ to its
universe $V$ = \{TRUE, FALSE\}.

However, to represent the MaxHorn2Sat instance \bf{M} as in
(\ref{eq:maxDef}) (or as in Theorem \ref{thom:maxhorn2sat} below), the
variables used will be
$X = \{x, y\}$, and the universe of the structure \bf{A} would be $Z$.
~Diagrammatically,
\begin{equation}
X = \{x, y\} ~\longrightarrow~ Z = \{z_1, z_2, z_3\} ~\longrightarrow~
V = \{TRUE, FALSE\}.
\end{equation}
\bf{A} maps (instantiates) $X$ to $Z$, and
\bf{B} maps (instantiates) $Z$ to $V$.

The second order variables $\mb{S}$ (to be used with \bf{A}) consists of
a single unary predicate $S$, that is $\mb{S} = \{S\}$ where $S$ is of
arity one.
$S$ can be considered as a guess of the map \bf{B}.
~In the above example, for a certain MaxHorn2Sat input clause, if
$\mb{A}(x) = z_1$,
$\mb{A}(y) = z_3$,
$S(x) = FALSE$ and
$S(y) = TRUE$,
then $S$ would have guessed that
$\mb{B}(z_1) = FALSE$ and
$\mb{B}(z_3) = TRUE$.

\subsection{The signature of A}\label{sec:sigA}

Henceforth, we shall work with the extended structure (\bf{A}, $S$),
where \bf{A} is a relational structure (the input).
$S$ is the only second order predicate, representing the output to the
optimisation problem.
$S$ is unary (that is, its arity is one);
it assigns true/false values to first order variables.

If variables $x$ and $y$ appear in a 2-literal MaxHorn2Sat clause, then
the clause can assume one of the following forms (and represented in the
signature of \bf{A} by the corresponding first order predicate on the
right):

\begin{center}
\begin{tabular}{|l|l|}
\hline
$\neg x \vee \neg y$ & $BothNeg (x,y)$ \\
\hline
$\neg x \vee y$ & $FirstNegSecondPos(x,y)$, or simply $FNSP(x,y)$ \\
\hline
$x \vee \neg y$ & $FirstPosSecondNeg(x,y)$, or simply $FPSN(x,y)$ \\
\hline
\end{tabular}
\end{center}

If a clause contains only one literal, insert a second literal and set
it to FALSE (explained in Sec. \ref{sec:oneLitClause}). 
We need two more predicates in the first order vocabulary:
$OnePos(x)$ and $OneNeg(x)$, depending on whether the literal is positive
or negative 
(explained in Sec. \ref{sec:oneLitClause}).

Hence the signature of \bf{A} consists of the following first order
predicates (all are binary):
$FNSP$, $FPSN$, $BothNeg$, $OnePos$ and $OneNeg$.

We need a few constants: $1 \le i \le 5$, or $I$ = \{1, 2, 3, 4, 5\}.
These will be used to indicate the type of clause
(explained in Sec. \ref{sec:clauseType}).  We also need a constant called
$NULL$, explained in Sec. \ref{sec:oneLitClause}.

The universe $U$ is the set of variables in the given MaxHorn2SAT
instance.  In the example above, $U$ = \{$z_1$, $z_2$, $z_3$\}.

\subsection{Counting satisfying clauses}

We make the following assumptions:
\begin{assum}
\begin{enumerate}
\item
A clause such as ($x \lor x$) is simplified to ($x$);

\item
Clauses such as ($x \lor \neg x$) are ignored;

\item
Assume that the list of variables is ordered.  For example, we can assume
that the variables have a certain sequence $x_1$ $<$ $x_2$ $<$ $\cdots$
$<$ $x_{n-1}$ $<$ $x_n$;

\item
In a two-variable clause consisting of different variables $x_i$ and
$x_j$, assume that $i < j$;

\item
We disallow duplication of clauses.
For example, if there are two equivalent clauses such as ($\neg x_i \vee
x_j$) and ($x_j \vee \neg x_i$) in the MaxHorn2Sat instance, where $i <
j$, we eliminate the clause $x_j \vee \neg x_i$, as per the previous
assumption;

\item
For any $(i,j)$ pair with $i < j$, distinct clauses such as ($x_i \lor
\neg x_j$) and ($\neg x_i \lor x_j$) \it{can} occur in the same
MaxHorn2Sat formula.
For the former clause, $FNSP(x_i, x_j)$ is true, and $FPSN(x_i, x_j)$ is
true for the latter.
\end{enumerate}
\label{assum:clauseTypes}
\end{assum}

Our approach is similar to that of Kolaitis-Thakur 1994 \cite{KT94},
where they provide an expression for the optimal value for Max3Sat
(optimisation version).

We only count satisfying MaxHorn2Sat clauses for the objective function.
That is, we count the number of tuples $(x, y, i)$ that satisfy $\phi
\land \tau \land \gamma$,
where $1 \le i \le 5$, and 
\begin{equation}
\phi \land \tau \land \gamma \equiv
\left(\bigvee_{i=1}^5 \phi_i \right) \land \tau \land \gamma.
\label{thePhees}
\end{equation}
The $\phi_i$'s are described below; $\tau$ and $\gamma$ are explained in
Sec. \ref{sec:clauseType}.

\subsection{Two-literal MaxHorn2Sat clauses}

Two-literal MaxHorn2Sat clauses can be satisfied in one of the following
ways:

$\ds \phi_1 \equiv FPSN(x,y) \wedge [S(x) \vee \neg S(y)]$.

$\ds \phi_2 \equiv FNSP(x,y) \wedge [\neg S(x) \vee S(y)]$.

$\ds \phi_3 \equiv BothNeg (x,y) \wedge [\neg S(x) \vee \neg S(y)]$.

\subsection{One-literal MaxHorn2Sat clauses}\label{sec:oneLitClause}

As mentioned earlier, convert one-literal clauses to two-literal
clauses.  (We do this, so that we can simply count the number of $(x,y,i)$
tuples that satisfy $\phi \land \tau \land \gamma$.)

If the literal is positive, then create a predicate called $OnePos(x)$,
create a constant called $NULL$,
and set the second literal $y$ to $NULL$, as if the clause is $x \vee y$;
The clause is true iff $x$ is true.

$\ds \phi_4 \equiv OnePos (x) \land (y = NULL) \land S(x)$.

Similarly if the literal is negative, then create $OneNeg(x)$:

$\ds \phi_5 \equiv OneNeg (x) \land (y = NULL) \land \neg S(x)$.

\subsection{The complete DNF formula}

The first two atoms in the definitions of $\phi_4$ and $\phi_5$ are first
order (known from the input).  Hence they can be combined into $D$ and
$E$ respectively, as below.

For convenience of writing, let us substitute

$A = FPSN(x, y)$, ~ 
$B = FNSP(x, y)$, ~ 
$C = BothNeg (x,y)$,

$D = OnePos (x) \land (y = NULL)$, ~ 
$E = OneNeg (x) \land (y = NULL)$,

$P = S(x)$, ~ 
$Q = S(y)$.

Then we can rewrite $\phi_i$ ($1 \le i \le 5$) as
\begin{equation}
\begin{array}{ll}
\phi_1 \equiv (A \wedge P) \vee (A \wedge \neg Q), &
\phi_2 \equiv (B \wedge \neg P) \vee (B \wedge Q), \\ [1mm]
\phi_3 \equiv (C \wedge \neg P) \vee (C \wedge \neg Q), ~ &
\phi_4 \equiv (D \wedge P), \\ [1mm]
\phi_5 = (E \wedge \neg P).
\end{array}
\label{phiDefine0}
\end{equation}

From (\ref{thePhees}), since one of the $\phi_i$'s should be satisfied
for a MaxHorn2Sat clause to be counted towards the objective function,
\begin{equation}
\begin{array}{rcl}
\phi & \equiv & \bigvee_{i=1}^5 ~ \phi_i \\ [1mm]
     & \equiv & (A \wedge P) \vee (A \wedge \neg Q) \vee
       (B \wedge \neg P) \vee (B \wedge Q) \\ [1mm]
    & & \vee (C \wedge \neg P) \vee (C \wedge \neg Q) 
         \vee (D \wedge P) \vee (E \wedge \neg P).
\end{array}
\label{phiDefine}
\end{equation}

Write $\phi \equiv k_1 \vee k_2 \vee \cdots \vee k_7 \vee k_8$,
corresponding to each of the 8 conjunct clauses above in $\phi$.

That is, $k_1 \equiv A \wedge P$, $k_2 \equiv A \wedge \neg Q$, $\cdots$, 
$k_7 \equiv D \wedge P$, and $k_8 \equiv E \wedge \neg P$.

\subsection{Converting DNF to CNF}

Now $\phi$ is in DNF, so we should convert it to CNF. ~Call the CNF form
as $\psi$ (or $\psi (x, y, S)$, to be more accurate).
There are 8 clauses in $\phi$ with 2 literals each, so $\psi$ will have
$2^8$ = 256 clauses\footnote{256 may be ``large", but still a finite number.},
 with 8 literals each --- one literal from each of the
8 clauses in $\phi$.
From (\ref{phiDefine}), we can write $\psi$ in lexicographic order as
\begin{equation}
\begin{array}{rcl}
\psi & \equiv & 
  (A \vee A \vee B \vee B \vee C \vee C \vee D \vee E) \wedge \cdots \\ [1mm]
& & \wedge (P \vee \neg Q \vee \neg P \vee Q \vee \neg P \vee \neg Q \vee P
\vee \neg P).
\end{array}
\label{psiDefine}
\end{equation}

We should ensure that each of the 256 disjunct clauses in $\psi$ is
Horn, which is what we do next.

\begin{lemma}
Each of the 256 clauses in $\psi$ is Horn.
\label{lem:128clauses}
\end{lemma}

\begin{proof}
Note that the literals $A$, $B$, $\cdots$, $E$ are first order (part of
the input), hence these do not affect the Horn condition; only the $P$'s
and $Q$'s and their negations do.

If there is an 8-literal clause in $\psi$ containing literals $P$ and
$\neg P$, it can be set to TRUE. ~Similarly for a clause containing $Q$
and $\neg Q$.

Anyway, we will run into trouble only if we have a clause $\psi_i$ in
$\psi$, that (i) contains literals $P$ and $Q$, and
(ii) contains neither $\neg P$ nor $\neg Q$.
~However, can such a clause evaluate to TRUE and hence can be
``ignored"?  Will such a clause obey the Horn condition?
The answer turns out to be yes.

There are only three ways in which we can come across a ``$P \vee Q$"
within a 8-literal clause of $\psi$:
\begin{itemize}
\item
Pick $P$ from $k_1$, $Q$ from $k_4$, and one of \{$A$, $B$, $\cdots$, $E$\}
from the other clauses, to obtain
$\psi_1 \equiv P \vee A \vee B \vee Q \vee C \vee C \vee D \vee E$.

This clause of $\psi$ contains $A$, $B$, $\cdots$ $E$ --- the five types
of clauses mentioned in (\ref{phiDefine0}), and one of them must occur;
they are mutually disjoint and collectively exhaustive.  So
$A \vee B \vee \cdots \vee E$ is (always) valid.  So $\psi_1$ can be set
to TRUE.

\item
Pick $P$ from $k_7$, $Q$ from $k_4$, and one of 
\{$A$, $B$, $\cdots$, $E$\} from the other clauses, to obtain
$\psi_2 \equiv A \vee A \vee B \vee Q \vee C \vee C \vee P \vee E$.

This clause only contains $A$, $B$, $C$ and $E$, but not $D$. ~However, we
know that $A \vee B \vee C \vee D \vee E$ is valid.  If $A$, $B$, $C$ and $E$
are false, then $D$ will be true (the $OnePos$ predicate) --- this
means, every clause in $\phi$ is false except $k_7$, which implies that
$P$ is true.  Hence $A \vee B \vee C \vee E \vee P$ is valid, which means
$\psi_2$ can be set to TRUE.

\item
Pick $P$ from $k_1$ and $k_7$, $Q$ from $k_4$, and one of \{$A$, $B$,
$\cdots$, $E$\} from the other clauses, to obtain
$\psi_3 \equiv P \vee A \vee B \vee Q \vee C \vee C \vee P \vee E$.
~ Apply the same argument as for $\psi_2$.  This sets $\psi_3$ to TRUE.
\end{itemize}

Hence each of the 256 clauses in $\psi$ is Horn.
\end{proof}

\subsection{The Type of Clause}\label{sec:clauseType}

We need a few more clauses to represent whether a certain
$(x, y)$ combination actually occurs in the MaxHorn2Sat formula, and in
which of the five $\phi_i$ (1 $\le i \le$ 5) varieties it occurs.
Furthermore, only clauses for which ($x < y$) or ($y = NULL$) should be
considered.  We express these as $\tau$:
\begin{equation}
\begin{array}{rcl}
\tau  & \equiv  & \bigwedge_{i=1}^6 ~ \tau_i, ~ \mbox{where}
\\ [1mm]
\tau_1  & \equiv  & (i = 1) \Leftrightarrow FPSN(x,y), \\ [1mm]
\tau_2  & \equiv  & (i = 2) \Leftrightarrow FNSP(x,y), \\ [1mm]
\tau_3  & \equiv  & (i = 3) \Leftrightarrow BothNeg(x,y), \\ [1mm]
\tau_4  & \equiv  & (i = 4) \Leftrightarrow OnePos(x) \land (y = NULL), \\ [1mm]
\tau_5  & \equiv  & (i = 5) \Leftrightarrow OneNeg(x) \land (y = NULL), ~ \mbox{and} \\ [1mm]
\tau_6  & \equiv  & (x < y)  \lor (y = NULL).
\end{array}
\end{equation}
Recall that we require $x < y$, in case $y \neq NULL$.
~But note that as per Assumption \ref{assum:clauseTypes} (Part 6), we
know that for the same $(x,y)$ pair, more than one value for $i$ is a
possibility.
Also, $\tau$ can be easily converted to CNF form.

However, everthing in $\tau$ is first order;
hence their truth values can be evaluated and substituted.
This does not affect the Horn condition.

But does a certain $(x,y,i)$ combination actually occur in the given
MaxHorn2Sat formula?  For instance, does $(x_4, x_7, 2)$ occur? That is,
does the clause $(\neg x_4 \lor x_7)$ occur?  For this,
we need another first-order predicate $\gamma(x,y,i)$; set this to true iff
the combination  $(x,y,i)$ occurs in the given input formula.
Furthermore, since $\gamma$ is first order, it does not affect the Horn
condition.

From all the arguments above including Lemma \ref{lem:128clauses}, we
conclude:

\begin{thom}
Let the structure $\mb{A}$ represent an instance of MaxHorn2Sat
defined in Problem \ref{prob:maxhorn2sat}.
Then the value of an optimal solution to $\mb{A}$ can be represented by
\begin{equation}
\O (\mb{A}) = \max_{S} \vert \{(x,y,i): (\mb{A}, S, x, y, i)
\models  ~ \psi(x, y, S) \land \tau \land \gamma(x,y,i) \}
\vert,
\label{maxDefPoptPB}
\end{equation}
where $\psi(x, y, S)$ is defined in
(\ref{psiDefine}), $x$ ranges over the universe $U$ (explained in Sec.
\ref{sec:sigA}), $y$ ranges over $U \cup \{NULL\}$, and the range for $i$
is $1 \le i \le 5$.
\label{thom:maxhorn2sat}
$\hfill \rule{2.0mm}{2.0mm}$
\end{thom}

Note that $\psi$ above is quantifier free ($\Pi_0$ or $\Sigma_0$ form).
This means

\bf{Corollary to Theorem \ref{thom:maxhorn2sat} and Discussion}:
Since it is known that MaxHorn2Sat is NP-hard, observe that even a
$\Pi_0$ Horn expression does not guarantee polynomial time solvability
for maximisation problems (assuming that \bf{P} $\not=$ \bf{NP}).

In \cite{cjtcs08}, it was shown that the MaxFlow$_{PB}$ problem (the
MaxFlow problem with unit weight edges) cannot be represented in Horn
$\Pi_0$ or Horn $\Sigma_1$ first order form; it needs a Horn
$\Pi_1$ sentence.  The optimal solution to this problem can be obtained
in polynomial time using Maximum Flow algorithms.

Hence it is unexpected that while a polynomially solvable problem,
MaxFlow$_{PB}$, has a Horn $\Pi_1$ lower bound, an NP-hard problem,
MaxHorn2Sat, can be expressed by a quantifier-free Horn sentence.

A similar anamoly was observed by Panconesi and Ranjan (1993)
\cite{pancoRanjan}: While the class MAX NP or MAX $\Sigma_1$ (defined in
Definition \ref{def:maxPi0andMaxNP}) can express NP-hard problems such as
Max3Sat, it is unable to express polynomially solvable problems such as
Maximum Matching.  This suggests that
\begin{conj}
\bf{Quantifier alternation} does not provide a precise characterisation of
computation time.  A hierarchy in quantifier alternation does not
translate to one in computation time.  We need to look at other
characteristics of logical formulae such as the number of variables, or
a combination of these.
\label{conj:hierarchyBreak}
\end{conj}

This section has further exposed the expressibility differences between
decision problems and optimisation problems.

% There may be a fluffy hierarchy between P and NP.

\section{Expressing optimality conditions with the help of duality}\label{sec:rescue}

From the question of logical expressibility of optimisation problems, we
next move to that of solving optimisation problems using Turing machines.

\bf{Recognizing (Verifying) Optimality}.  In general, the question,
~\it{Given a solution \bf{T} to an instance \bf{A} of an optimisation
problem Q, is it an optimal solution?}~ 
is as hard to answer as determining an optimal solution, necessitating a
$\Sigma_2$ second order\footnote{Defined in Definition
\ref{def:pie1pie2sigma1sigma2}.}
sentence as in (\ref{lazyOpt}) below.
However, under certain conditions, such as when the duality gap is zero,
optimal solutions can be recognised more efficiently, and can be
expressed in existential second order (ESO, or second order $\Sigma_1$)
logic.

\it{Duality Gap} is the difference between the optimal solution values
for the primal and dual problems; these two problems are defined below in
(\ref{eq:problemP1}) and (\ref{eq:problemP2}).
For problems such as LP and MaxFlow-MinCut, the duality gap has been
shown to be zero; that is, they posess the \it{strong duality} property.
However, for other problems such as Integer Programming, there is no
known dual problem that guarantees strong duality; hence expressions
that capture the simultaneous existence of primal and dual optimal
solutions with equal value
(such as (\ref{LPcharac}) and (\ref{MFMCcharac}))
cannot be derived, at least until a dual that guarantees strong duality
is discovered.

The above question can also be phrased as a classical decision problem
(for maximisation):
\it{Given a solution \bf{T} for an instance \bf{A} with solution value
f(\bf{T}), is there another solution \bf{S} such that f(\bf{S}) $>$
f(\bf{T})?}

An optimal solution \bf{T} to an instance \bf{A} of an optimisation
problem $Q$ can easily be represented as the best among all feasible
solutions \bf{S}:
\begin{equation}
\exists \mb{T} \forall \mb{S} ~ \phi(\mb{A}, \mb{T}) \land 
\phi(\mb{A}, \mb{S}) \wedge [f(\mb{A}, \mb{T}) \ge
f(\mb{A}, \mb{S})],
\label{lazyOpt}
\end{equation}
where $\phi$ represents satisfaction of the constraints to \bf{A}, and
$f$ is the objective function referred to, in Definitions
\ref{def:pOptProblem} and \ref{def:npOptProblem}.
The formula $\phi$ captures the constraints, such as
$\mb{g(x)} = \mb{b}$ and $\mb{h(x)} \le \mb{c}$ in (\ref{eq:problemP1})
below.
~ $\mb{g(x)}$ and $\mb{h(x)}$ are functions of $\mb{x} \in \bb{R}^n$.

[Note that the above formula represents an optimal solution to a
maximisation problem; we can write a similar formula for minimisation;
simply change the last condition to
$f(\mb{A}, \mb{T}) \le f(\mb{A}, \mb{S})$.]

Recall that a maximisation problem $P_1$ in the $\bb{R}^n$ Euclidean
space can be represented as follows \cite{bazaraaEtAl}:
\begin{equation}
\begin{array}{lrl}
& \mbox{Maximise} & f_1(\mb{x}):\bb{R}^n \rightarrow \bb{R}, \\ [1mm]
(P_1)~~ & \mbox{subject to} & \mb{g(x)} = \mb{b}, 
        ~~\mb{h(x)} \le \mb{c}, \\ [1mm]
& \mbox{where} & \mb{x} \in \bb{R}^n, 
\mb{b} \in \bb{R}^{m_1} \mbox{ and }
\mb{c} \in \bb{R}^{m_2}.
\end{array}
\label{eq:problemP1}
\end{equation}

For several optimisation problems, an optimal solution can be recognised
when a feasible solution obeys certain \it{optimality} conditions.  In
such cases, it is unnecessary to represent an optimal solution \bf{T} as
in (\ref{lazyOpt}).  The duality concept in optimisation can play an
important role here.

Let $\mb{u} \in \bb{R}^{m_1}$ and $\mb{v} \in \bb{R}^{m_2}$ be two
vectors of variables with $\mb{v} \ge \mb{0}$.
Given a \it{primal} problem $P_1$ as in (\ref{eq:problemP1}), its
Lagrangian dual problem $P_2$ can be represented as (see
\cite{bazaraaEtAl}):
\begin{equation}
\begin{array}{lrl}
& \mbox{Minimise} & \theta(\mb{u}, \mb{v}) \\ [1mm]
(P_2)~~ & \mbox{subject to} & \mb{v} \ge \mb{0}, \\ [1mm]
& \mbox{where} & 
\theta(\mb{u}, \mb{v}) = 
\inf_{\mb{x} \in  \bb{R}^n}
 ~ \{f(\mb{x}) + \sum_{i = 1}^{m_1} u_i g_i (\mb{x}) +
                 \sum_{j = 1}^{m_2} v_j h_j (\mb{x}) \}.
\end{array}
\label{eq:problemP2}
\end{equation}
Furthermore, 
$\ds g_i \mb{(x)} = b_i$ [$\ds h_j \mb{(x)} \le c_j$] is the 
$i^{th}$ equality [$j^{th}$ inequality] constraint respectively.

We have demonstrated $\Sigma_1$ (i.e. with existential quantifier) second
order expressibility using Lagrangian duality in the following sections.
However, other types of duality may be used, such as Fenchel duality or
the geometric duality or the canonical duality, as long as they provide a
zero duality gap, and optimality conditions that can be verified
efficiently (say, in polynomial time).

\subsection{Computational models}\label{sec:compModels}

Turing machine (TM) based computational models for solving an
optimisation problem $Q$ come in two flavours:

\bf{Model 1}. The input consists of a problem instance such as in
(\ref{eq:problemP1}). 
If the instance has a feasible solution, the output is a string
representing an optimal solution; otherwise, the TM crashes (no output).
Corresponding to the class P in the world of decision problems, the
class here is \bf{FP} (\cite{papa}, Page 230).

(However, in the case of decision problems that are in the class NP and
the optimisation problems that are NP-hard, the correspondence between
NP and \bf{FNP} is not exact.)

\bf{Model 2}.
In addition to a problem instance such as in (\ref{eq:problemP1}), the
input consists of a parameter $K$, which is a bound on the optimal
solution value.  The TM is a ``decision" machine, that is, one whose
output is simply a yes or a no; call this machine as $M_1$.
~The method then to solve $Q$, by a Turing machine, say $M_2$, is to do
a binary search on solution values, calling $M_1$ a logarithmic ($\log
U$) number of times, where $U$ is an upper bound on the optimal solution
value.
Thus we make a \it{weakly polynomial}\footnote{For a graph problem, an
algorithm is strongly polynomial if the running time is a polynomial in
the number of vertices and/or edges; it becomes weakly polynomial if the
running time is a polynomial in the logarithm of edge weights.
In Linear Programming, this translates to the number of
variables/constraints versus the data in the coefficient matrix \bf{A}
and the right side vector \bf{b}.
~In the graph problem, the number of vertices/edges represents the
\bf{number} of input parameters, whereas the edge weights represent the
\bf{values} of such parameters.}
number of calls to $M_1$.
Each call to $M_1$ involves answering a question such as:
``Is there a feasible solution \bf{S} satisfying the constraints, such
that the objective function value $f(\mb{A}, \mb{S})$ is greater than or
equal to $K$?", for a maximisation problem.

We make a few assumptions here:
\newline
(a) All feasible solutions have non-negative values;
\newline
(b) Given a solution \bf{x}, $f_1 (\mb{x})$, $\mb{g(x)}$ and $\mb{h(x)}$
can be computed in time polynomial in the size of \bf{x}; and 
\newline
(c) Given an input for an instance of (\ref{eq:problemP1}), which
consists of parameters for the three functions $f_1 (\mb{x})$,
$\mb{g(x)}$ and $\mb{h(x)}$, as well as \bf{b} and \bf{c}, an upper bound
$U$ on the optimal solution value $V$ can be computed within time
polynomial in the size of these input parameters. 
\newline
(In cases where it is not possible to compute $V$ efficiently, we need a
simple upper bound $U$ that can be quickly computed.)

For more details on such TM models, the reader is referred to
Papadimitriou \cite{papa}.

If the problem answered by $M_1$ is in the NP class, then the complexity
of solving $Q$ is in $\ds P^{NP}$, since $M_2$ makes a
polynomial number of calls to the oracle $M_1$ (and $M_1$ solves a
problem in NP).

Similarly, if the problem answered by $M_1$ is in the class P, then 
the complexity of solving $Q$ is in $\ds P^{P}$, which is simply P
(although strictly speaking, this is weakly polynomial due to the $\log
V$ number of calls).

The method used in Model 2, binary search, has been recognised/adopted
for solving optimisation problems since the discovery of the class NP.
~It involves making a polynomial number of calls to a ``decision TM" (a
TM that solves decision problems).

However, we show in this section that for pairs of problems with a
duality gap of zero, a single call to a decision TM is sufficient.
If the machine answers yes, then the primal and the dual problems have
optimal solutions; otherwise, neither problem has an optimal
solution (at least one of the problems will be infeasible, and one of
them may have an unbounded optimal solution).
This is demonstrated by second order $\Sigma_1$ sentences such as
(\ref{LPcharac}) and (\ref{MFMCcharac}), which implies, as per Fagin's
result below, that such a machine produces a yes/no answer in NP time.
\begin{thom}
\cite{fagin}
A decision problem can be logically expressed in ESO if and only if
it is in NP.
\label{thom:fagin}
\end{thom}

The following theorem is the deterministic counterpart of Fagin's
result. 
It characterises {P} as the class of decision problems definable by ESO
universal Horn formulae.

\begin{thom}
(Gr\"{a}del \cite{gradel91})
For any ESO Horn expression as defined in Definition \ref{def:ESOhorn}, the
corresponding decision problem is a member of {P}.

The converse is also true --- if a problem $\mathcal{P}$ is a member of {P},
then it can be expressed in ESO Horn form  --- but only if a successor
relation (defined in Def. \ref{def:successor}) is allowed to be included
in the vocabulary of the first-order formula $\psi$ (of Def.
\ref{def:ESOhorn}). 
% $\hfill \rule{2.0mm}{2.0mm}$
\label{thom:decision}
\end{thom}

\begin{rem}
The polynomial time computability in the first part of Theorem 
\ref{thom:decision} is due to the fact that the first order part of
formulae  representing decision problems can be reduced to propositional
Horn formulae, which can be solved in time linear in the number of
predicates which are second order and unknown (that is, not a part of the
input)\footnote{The first order predicates are part of the input, hence
their truth values are known and can be substituted.}.
\label{rem:hornPolySoln}
\end{rem}

\subsection{Linear programming}\label{sec:LP}

Linear Programming (LP) is known to be in the class P
\cite{fangPuthenpura}, whereas Integer Programming (IP) is NP-hard
\cite{gj}.

In the case of Linear Programming (LP), using Lagrangian
duality, the primal and dual
problems $P_3$ and $P_4$ respectively, can be stated as follows:
\begin{equation}
\begin{array}{crlccrl}
 (P_3) & \mbox{Maximise} & f_1(\mb{x}) = \mb{c}^T \mb{x}, 
 & & (P_4) & \mbox{Minimise} & f_2(\mb{y}) = \mb{b}^T \mb{y},
\\ [1mm]
& \mbox{subject to} & \mb{Ax} \le \mb{b}, ~ \mb{x} \ge \mb{0}, 
 & & & \mbox{subject to} & \mb{A}^T \mb{y} \ge \mb{c}, ~ \mb{y} \ge \mb{0},
\\ [1mm]
& \mbox{where} & \mb{x}, \mb{c} \in \bb{R}^n, &  &  &
 \mbox{and} 
&  \mb{y}, \mb{b} \in \bb{R}^{m},
\end{array}
\label{LPdual}
\end{equation}
after the usual process \cite{hadley} of converting unrestricted
variables (if any) to non-negative variables, and equality constraints
(if any) to inequality constraints, in the primal problem.
Here, $y_i$ ($x_j$) is the $i^{th}$ dual ($j^{th}$ primal) variable
corresponding to the $i^{th}$ primal ($j^{th}$ dual) constraint. 
When the primal and dual problems have feasible solutions, then
they both have optimal solutions 
$\ds \mb{x^*} = (x_1^*, x_2^*, \cdots, x_n^*)$ and 
$\ds \mb{y^*} = (y_1^*, y_2^*, \cdots, y_m^*)$
such that
the two objective functions are equal:
$\ds \mb{c}^T \mb{x^*} = \mb{b}^T \mb{y^*}$.
(Almost every book on LP should explain this result.  See for example,
\cite{hadley}.)

For LP's, the \it{complementary slackness} conditions below are known to
be necessary and sufficient conditions for the existence of an optimal
primal solution \bf{S} and an optimal dual solution \bf{T}:
\begin{eqnarray}
y_i^{\ast} (b_i - A_i \mb{x^{\ast}})  =  0, & y_i^{\ast} \ge 0, &
b_i - A_i \mb{x^{\ast}} \ge 0, ~~ i \in \{1, 2, \cdots, m\} 
\label{compSlack1} \\ [1mm]
x_j^{\ast} (c_j - A_j^T \mb{y^{\ast}})  = 0, & x_j^{\ast} \ge 0, &
c_j - A_j^T \mb{y^{\ast}} \ge 0, ~~ j \in \{1, 2, \cdots, n\} 
\label{compSlack2}
\end{eqnarray}
where $A_i$ is the $i^{th}$ row of \bf{A},
$A_j^T$ is the $j^{th}$ column of \bf{A},
$(b_i - A_i \mb{x^{\ast}})  =  0$ is derived from the $i^{th}$ primal
constraint, and $(c_j - A_j^T \mb{y^{\ast}})  = 0$ is derived from the
$j^{th}$ dual constraint.

Thus the existence of \bf{S} and \bf{T} can be expressed as
\begin{equation}
\exists \mb{S} \exists \mb{T} 
~ [\forall i ~ \psi_1(i)] \wedge [\forall j ~ \psi_2(j)] \wedge 
\phi_p (\mb{S}) \wedge \phi_d (\mb{T}),
\label{LPcharac}
\end{equation}
where $\psi_1(i)$ [$\psi_2(j)$] logically captures the $i^{th}$
[$j^{th}$] constraint in (\ref{compSlack1}) [(\ref{compSlack2})]
respectively.
Also, $\phi_p$ and $\phi_d$ model the primal and dual constraints in
(\ref{LPdual}) respectively.

We are not concerned about the first order part of the above
expression, 
$\ds [\forall i ~ \psi_1(i)] \wedge [\forall j ~ \psi_2(j)] \wedge 
\phi_p \wedge \phi_d$.
What is of interest to us is that the existence of optimal solutions for
the primal and dual problems can be expressed in ESO, existential second
order logic; 
a $\Sigma_2$ second order sentence as in (\ref{lazyOpt}) 
is unnecessary.

Note that (\ref{LPcharac}) returns neither an optimal cost nor an optimal
solution; this is consistent with Theorem \ref{thom:decision}.
Providing a framework to compute these entities is \it{not} our concern
at this juncture.

\begin{rem}
Applying Theorem \ref{thom:fagin},
it follows that recognition of an optimal solution, for certain problems
that obey strong duality (such as LP), is in the computational class NP.

(One could argue that the existence of a feasible
solution\footnote{A word of caution --- \bf{Feasible solutions}, a
difference in terminology: Fagin and Gr\"{a}del \cite{gradel91} have
syntactically characterised \it{feasible} solutions for classes NP and P
respectively.  However the ``feasibility" captured by an ESO expression,
as described by Fagin and Gr\"{a}del, also includes an upper (lower)
bound on the objective function of a minimisation (maximisation)
problem, such as $f_1(\mb{x}) \ge K$ where $K$ is a constant --- not
just satisfaction of the constraints such as \bf{Ax} $\le$ \bf{b},
\bf{x} $\ge$ \bf{0} in (\ref{LPdual}).
In this paper, we differ from this view; when we talk about feasibility,
we only refer to satisfaction of constraints such as \bf{Ax} $\le$
\bf{b}, \bf{x} $\ge$ \bf{0}.}
for an optimisation problem, satisfying constraints such as $\mb{Ax} \le
\mb{b}, ~ \mb{x} \ge \mb{0}$, implies the existence of an optimal
solution.)
\end{rem}

\subsection{Polynomially solvable problems}\label{sec:polytime}
But what if the primal and dual problems are polynomially solvable?
Can this be reflected in expressions such as (\ref{LPcharac})?
The answer turns out to be yes --- well, at least for Linear
Programming.
Recall from Theorem (\ref{thom:decision}) that to express polynomial
solvability, the first order part of (\ref{LPcharac}) needs to be a
universal Horn formula, when the underlying input structure has a
built-in successor relation.

The theory of \it{Interior Point} methods \cite{fangPuthenpura} imply the
polynomial solvability of the primal and the dual problems.
From this and Theorem \ref{thom:decision}, it follows that $\phi_p$ and
$\phi_d$ can be expressed as universal Horn formulae, as long as the
underlying structure \bf{B} obeys the conditions of Theorem
\ref{thom:decision} (such as the signature of \bf{B}).

As for the complementary slackness conditions (\ref{compSlack1}) and
(\ref{compSlack2}),  we only need to express 
$\ds y_i^{\ast} (b_i - A_i \mb{x^{\ast}}) =  0$ and 
$\ds x_j^{\ast} (c_j - A_j^T \mb{y^{\ast}}) = 0$, 
since the other conditions have been expressed in 
$\phi_p$ and $\phi_d$.

$\ds y_i^{\ast} (b_i - A_i \mb{x^{\ast}}) =  0$ can be expressed as 
$\ds \psi_1(i) \equiv Y(i) \vee B\_A(i, X)$, where $Y(i)$ is a predicate
which is true iff $\ds y_i^{\ast} = 0$, and 
$\ds B\_A(i, X)$ is a predicate which is true iff 
$\ds b_i - A_i \mb{x^{\ast}} =  0$.
The formula $\psi_1(i)$ is not Horn.  However, since $\ds y_i^{\ast} = 0$
and $\ds b_i - A_i \mb{x^{\ast}} =  0$ do not occur anywhere else in 
(\ref{LPcharac}), we can negate the predicates and modify $\ds \psi_1(i)$.

As in Theorem \ref{thom:hornMax}, the Horn condition in the formula
$\eta$ applies only to the second order predicates in \bf{S} and \bf{T}.
~In this case, it applies to predicates that involve unknowns such as
$x_j$ and $y_i$.

Let $\ds YnotEq0(i)$ be true iff $\ds y_i^{\ast} \not= 0$, and 
$\ds B\_AnotEq0(i, X)$ be true iff $\ds b_i - A_i \mb{x^{\ast}} \not=
0$.
Using these, one can rewrite $\psi_1(i)$ as
\begin{equation}
\psi_1(i) \equiv \neg YnotEq0(i) \vee \neg B\_AnotEq0(i, X),
\end{equation}
which is a Horn formula.

$YnotEq0$ and $B_AnotEq0$ can be constructed in polynomial time.
The predicate $B_AnotEq0$ is more crucial here, since it involves $b_i$,
$A_i$ and $X$.  But checking this is polynomial, since we mainly need to
compute a dot product of the row $A_i$ with $X$.
The logic machinery needed to express the arithmetic can be built into
the first order vocabulary (for example, see the first chapter of
Immerman's book), such that these FO predicates are not affected by the
Horn condition.

Similarly, the formula $\psi_2(j)$ in
(\ref{LPcharac}) can be expressed in Horn form:
\begin{equation}
\psi_2(j) \equiv \neg XnotEq0(j) \vee \neg C\_AnotEq0(j, Y).
\end{equation}
Now that we know that all four subformulae in the first order part of 
(\ref{LPcharac}) can be expressed in universal Horn form, we can
conclude that the formula in (\ref{LPcharac}) fully obeys the conditions
of Theorem \ref{thom:decision}; that is, ESO logic with the first order
part being a universal Horn formulae (that is, the quantifier-free part is
a conjunction of Horn clauses).  Hence we can state that
\begin{thom}
For a pair of primal and dual Linear Programming problems as in 
(\ref{LPdual}), and hence obeying strong duality, when the underlying
input structure has a built-in successor relation, the existence of
optimal solutions for the primal and the dual can be expressed in ESO
logic with the first order part being a universal Horn formula, and the
optimal solutions can be computed in polynomial time (a) using the
technique in Remark \ref{rem:hornPolySoln},
and
(b) by a single call to a decision Turing machine (which returns yes/no
answers).
\end{thom}

But does strong duality imply polynomial time solvability?  This is the
subject of another manuscript \cite{manyem10}.

\subsection{Maxflow mincut}\label{sec:maxFlowMinCut}

The MaxFlow-MinCut Theorem is another example where Lagrangian duality
plays an important role in characterizing optimal solutions.
The MaxFlow and MinCut problems are dual to each other.
At optimality, the values of the two optimal solutions coincide.
An optimal solution to MaxFlow can be syntactically recognised by an
``optimality condition", rather than a comparison of the objective
function value with those of all other feasible solutions. 

The MaxFlow and the MinCut problems have been defined in several books.
For example, see \cite{ahujaMagOrlin} or \cite{hadley}.
The decision versions of both problems are known to be in the complexity
class P.
~We reproduce the definitions below for convenience.
\begin{defn}
The \bf{MaxFlow} problem: 
\newline
{\em Given}.
We are given a network $G = (V, E)$ with 2 special vertices $s, t \in
V$, $E$ is a set of directed edges, and each edge $(i,j) \in E$ has a
capacity $C_{ij} > 0$.
\newline
{\em To Do}.
Determine the maximum amount of flow that can be sent from $s$ to $t$
such that in each edge $(i,j) \in E$, the flow $f(i,j)$ is at most its
capacity $C_{ij}$. That is, $0 \le f(i,j) \le C_{ij}, ~ \forall (i,j) \in E$.
\label{maxFlowDefn}
\end{defn}

An \bf{S-T Cut} is a non-empty subset $U$ of $V$ such that $S \in U$ and
$T \in \bar{U}$, where $\bar{U} = V - U$.
~[If $U$ is used as a second order predicate, 
then $U(i)$ is true for all vertices $i \in U$;
it follows that $U(S)$ is true and $U(T)$ is false;]
~The capacity of the cut, written as $C(U)$, is the sum of the
capacities of all edges $(i,j)$ such that $i \in U$ and $j \in \bar{U}$:
\[
C(U) = \sum_{(i,j) \in E, ~ i \in U, ~ j \in \bar{U}} C_{ij}.
\]

\begin{defn}
The \bf{MinCut} problem:
\newline
{\em Given}.
Same as the MaxFlow problem.
\newline
{\em To Do}.
Of all the $S-T$ cuts in $G$, find a least cut; that is, a cut with the
least capacity.
\label{minCutDefn}
\end{defn}

The \it{optimality condition} for the MaxFlow problem is that
there exists a least $S-T$ cut, $U$, such that 
\begin{itemize}
\item
(forward direction)
For every edge $(i,j)$ in the edge set $E$ such that $i \in U$ and $j
\in \bar{U}$, the flow in $(i,j)$, $f(i,j)$, is equal to its capacity
$C_{ij}$;
\item
(backward direction)
For every edge $(i,j) \in E$ such that $i \in \bar{U}$ and $j
\in U$, $f(i,j) = 0$; and
\item
The maximum flow, that is, the optimal solution value for the MaxFlow
problem, is equal to $C(U)$, the capacity of the cut $U$.
\end{itemize}

This condition can be syntactically characterised as
\begin{equation}
\begin{array}{ll}
\exists {U} \exists {F} ~ \forall i \forall j ~ U(S) \wedge \neg U(T) 
\\ [1mm]
\wedge ~ [E(i,j) \wedge U(i) \wedge \neg U(j) \longrightarrow F(i,j,C_{ij})]
\\ [1mm]
\wedge ~ [E(i,j) \wedge \neg U(i) \wedge U(j) \longrightarrow F(i,j,0)]
~ \wedge ~ \psi, ~ \mbox{where}
\label{MFMCcharac}
\end{array}
\end{equation}
$U$ and $F$ are second order predicates; 
\newline
$E(i,j)$ is a first order relation which is true whenever $(i,j)$ is an
edge in the input graph;
\newline
$U(i)$ is true when $i \in$ vertex set $U$; 
\newline
$F(i,j,v)$ is true when the flow in the edge $(i,j)$ equals $v$; and
\newline
$\psi$ models the flow conservation constraint at all nodes.

The flow conservation constraint is a necessary
constraint for the MaxFlow problem (decision version), which is known to
be polynomially solvable.
Hence as per Theorem \ref{thom:decision}, we can express $\psi$ in ESO
universal Horn logic.

Once more, by exploiting previously proven optimality conditions (the
MaxFlow MinCut theorem in this case), we
have been able to characterise the primal optimal solution $F$ and the dual
optimal solution $U$, in existential second order logic (ESO).

Similarly in Convex Programming, the Karush-Kuhn-Tucker conditions
provide sufficient conditions for the optimality of a feasible solution.

\section{Effect of zero duality gap}

Theorem \ref{thom:decision} provides an expression for the existence of a
feasible solution and polynomial time computation (Remark
\ref{rem:hornPolySoln}). 
What we present here is an improvement on that result, for problems that
obey strong duality.

Observe that expressions such as those in (\ref{LPcharac}) and
(\ref{MFMCcharac}) are possible only if there is no \it{duality
gap}, that is, when the duality gap is zero.
The primal optimality condition implies dual feasibility and vice versa.

% (ARE THE FORMULAE IN 
% (\ref{LPcharac}) and (\ref{MFMCcharac}) Horn?  If so, using the
% ``propositional Horn formula in linear time" property, can you get
% something algorithmic out of it?)

To our knowledge, all known problem-pairs with a zero duality gap, also
known as \it{strong duality}, are polynomially solvable (a well-known
exception is Semidefinite Programming, see \cite{ramana}).
The decision versions of all such optimisation problems can be shown to
be in the complexity class NP $\cap$ CoNP \cite{manyem10}.
The problem class P is closed under complementation; that is, P =CoP
\cite{papa}.

Problems in NP $\cap$ CoNP can be expressed in both ESO and USO
(universal second order logic), since USO precisely characterises
problems in CoNP.

% (How does all this relate to the POLYNOMIAL HIERARCHY?
% Saddle Points? (Linear) complementarity problems?)

\subsection{Problems that obey weak duality but not strong duality}

It is known that universal second order (USO) logic precisely
characterises problems in CoNP \cite{EF99}.
Let $\Phi$ be the formula
\begin{equation}
\Phi \equiv \forall \mb{S} \forall \mb{T} 
~ [g(\mb{T}) < f(\mb{S})] \wedge 
\phi_p (\mb{S}) \wedge \phi_d (\mb{T}),
\label{weakDual}
\end{equation}
where $\phi_p$ ($\phi_d$) model the primal (dual) constraints
respectively;
$f(\mb{S})$ and $g(\mb{T})$ represent the primal and dual objective
function values.
The relation $[g(\mb{T}) < f(\mb{S})]$ for all ($\mb{S}$, $\mb{T}$)
pairs implies that weak duality is obeyed, but not strong duality.

\section{Conclusions}

In this manuscript, we have shown that while all polynomially solvable
decision problems can be expressed as universal ($\Pi_1$) Horn sentences,
if $P \neq NP$, optimisation problems defy such a characterisation, in
the framework defined in expression (\ref{eq:maxDef}).
We showed this by demonstrating that even a $\Pi_0$ Horn formula is
unable to guarantee polynomial time solvability (assuming that P $\not=$
NP).
~In addition, by connecting descriptive complexity with optimisation
duality, we have shown how a certain class of optimisation problems can be
solved by a single call to a decision Turing machine, and presented two
examples.
What we have shown here may just be the beginning; exploring relationships
between duality and finite model theory could provide more interesting
results in complexity theory.

\bf{Acknowledgements}.
I thank James Gate and Iain Stewart at the University of Durham (UK) for
motivating me towards this line of research.
A part of this work was carried out while I was visiting the National
Cheng Kung University (NCKU) in Taiwan on a visiting fellowship;
support from NCKU is gratefully acknowledged.
Research also supported by grants from the National Natural Science
Foundation of China (No. 11071158) and the Key Disciplines of Shanghai
Municipality (No. S30104).

\end{document}